\documentclass[11pt]{article}

\usepackage{amsmath}
\usepackage{amssymb,amsfonts,textcomp}
\usepackage{latexsym}
\usepackage[numbers]{natbib}
\usepackage[noend,ruled]{algorithm2e}

\usepackage{graphicx}
\usepackage{latexsym}
\usepackage{multirow}
\usepackage{rotating}

\newcommand{\np}{{{{\mathrm{NP}}}}}
\newcommand{\p}{{{{\mathrm{P}}}}}

\newcommand{\rot}{{{\mathrm{R}}}}
\newcommand{\poly}{{{\mathrm{poly}}}}

\newtheorem{theorem}{Theorem}
\newtheorem{definition}[theorem]{Definition}
\newtheorem{lemma}[theorem]{Lemma}
\newtheorem{proposition}[theorem]{Proposition}
\newtheorem{corollary}[theorem]{Corollary}
\newtheorem{example}[theorem]{Example}

\newenvironment{proof}{\noindent\textbf{Proof}\quad}{\hfill$\Box$\medskip}
\newenvironment{proofof}[1]{\noindent\textbf{#1}\quad}{\hfill$\Box$\medskip}
\newenvironment{proofsketch}{\noindent\textbf{Proof sketch}\quad}{\hfill$\Box$\medskip}

\newcommand{\fixlist}{\addtolength{\itemsep}{-4pt}}
\newcommand{\omitforIJCAI}[1]{}

\newcommand{\ordset}[1]{{#1_\rightarrow}}
\newcommand{\revset}[1]{{#1_\leftarrow}}

\newcommand{\pos}{{{{\mathrm{pos}}}}}
\newcommand{\borda}{{{\mathrm{B}}}}

\newcommand{\alphaborda}{{{\alpha_\borda}}}

\renewcommand{\|}{{{|}}}

\title{The Complexity of Fully Proportional Representation for Single-Crossing Electorates}

\author{\makebox[.3\linewidth]{Piotr Skowron}  \\ University of Warsaw \\ Warsaw, Poland \\ p.skowron@mimuw.edu.pl
\and
       \makebox[.4\linewidth]{Lan Yu} \\ Nanyang Technological University \\ Singapore \\ jen.lan.yu@gmail.com
\and
        \makebox[.3\linewidth]{Piotr Faliszewski} \\ AGH University \\ Krakow, Poland \\ faliszew@agh.edu.pl
\and 
        \makebox[.4\linewidth]{Edith Elkind} \\ Nanyang Technological University \\ Singapore \\ eelkind@ntu.edu.sg
}

\begin{document}
\maketitle

\begin{abstract}
  We study the complexity of winner determination in single-crossing
  elections under two classic fully proportional representation
  rules---Chamberlin--Courant's rule and Monroe's rule. Winner
  determination for these rules is known to be NP-hard for
  unrestricted preferences. We show that for single-crossing
  preferences this problem admits a polynomial-time algorithm for
  Chamberlin--Courant's rule, but remains NP-hard for Monroe's
  rule. Our algorithm for Chamberlin--Courant's rule can be modified
  to work for elections with {\em bounded single-crossing width}.  To
  circumvent the hardness result for Monroe's rule, we consider
  single-crossing elections that satisfy an additional constraint,
  namely, ones where each candidate is ranked first by at least one
  voter (such elections are called narcissistic).  For single-crossing
  narcissistic elections, we provide an efficient algorithm for the
  egalitarian version of Monroe's rule.
\end{abstract}

\section{Introduction}\label{sec:intro}
Parliamentary elections, i.e., procedures for selecting a fixed-size
set of candidates that, in some sense, best represent the voters,
received a lot of attention in the literature.  Some well-known
approaches include first-past-the-post system (FPTP), where the voters
are divided into districts and in each district a plurality election
is held to find this district's representative; party-list systems,
where the voters vote for parties and later the parties distribute the
seats among their members; SNTV (single nontransferable vote) and Bloc
rules, where the voters cast $t$-approval ballots and the rule picks
$k$ candidates with the highest approval scores (here $k$ is the
target parliament size, and $t = 1$ for SNTV and $t = k$ for Bloc);
and a variant of STV (single transferable vote). In this paper, we
focus on two voting rules that, for each voter, explicitly define the
candidate that will represent her in the parliament (such rules are
said to provide {\em fully proportional representation}), namely,
Chamberlin--Courant's rule~\cite{cha-cou:j:cc} and Monroe's
rule~\cite{mon:j:monroe}.  Besides parliamentary elections, the winner
determination algorithms for these rules can also be used for other
applications, such as resource
allocation~\cite{mon:j:monroe,sko-fal-sli:c:multiwinner} and
recommender systems~\cite{bou-lu:c:chamberlin-courant}.

Let us consider an election where we seek a $k$-member parliament
chosen out of $m$ candidates by $n$ voters. Both Chamberlin--Courant's and Monroe's rule
work by finding a function $\Phi$ that
assigns to each voter $v$ the candidate that is to represent $v$ in
the parliament. This function is required to assign at most $k$ candidates altogether.\footnote{Under Monroe's rule we are required
  to pick exactly $k$ winners. Some authors also impose this requirement 
  in the case of Chamberlin--Courant's rule, but allowing for smaller parliaments
  appears to be more consistent with the spirit of this rule and  
  is standard in its computational analysis
  (see, e.g.,~\cite{bou-lu:c:chamberlin-courant,bet-sli-uhl:j:mon-cc,sko-fal-sli:c:multiwinner,sko-fal-sli:c:monroe-cc-experiments,yu-chan-elk:c:sptrees}).  
  In any case, this distinction has no bite if there are at least
  $k$ candidates that are ranked first by some voter, which is usually the case
  in political elections.
}
Further, under Monroe's rule each candidate is either assigned to about
$\frac{n}{k}$ voters or to none. The latter restriction does not apply
to Chamberlin--Courant's rule, where each selected candidate may represent
an arbitrary number of voters, and, as a consequence, the parliament elected
in this manner may have to use weighted voting in its proceedings.
Finally, each voter should be represented by a candidate that this
voter ranks as high as possible.

To specify the last requirement formally, we assume that there is a
global \emph{dissatisfaction function} $\alpha$, $\alpha \colon
\mathbb{N} \rightarrow \mathbb{N}$, such that $\alpha(i)$ is a voter's
dissatisfaction from being represented by a candidate that she views
as $i$-th best. (A typical example is Borda dissatisfaction
function $\alpha_\borda$ given by $\alpha_\borda(i) = i-1$.) 
In the utilitarian variants of Chamberlin--Courant's and Monroe's rules we seek assignments that
minimize the sum of voters' dissatisfactions; in the egalitarian
variants (introduced recently by Betzler et al.~\cite{bet-sli-uhl:j:mon-cc})
we seek assignments that minimize the dissatisfaction of the
worst-off voter.

Chamberlin--Courant's and Monroe's rules have a number of attractive properties, 
which distinguish them from other multiwinner rules. 
Indeed, they elect parliaments that (at least
in some sense) proportionally represent the voters, ensure that
candidates who are not individually popular cannot make it to the
parliament even if they come from very popular parties, and take
minority candidates into account. In contrast, FPTP can
provide largely disproportionate results, party-list systems
cause members of parliament to feel more
responsible to the parties than to the voters, SNTV and Bloc tend
to disregard minority candidates, and STV is believed to
put too much emphasis on voters' top preferences.

Unfortunately, Chamberlin--Courant's and Monroe's rules do have one
flaw that makes them impractical: It is $\np$-hard to compute their
winners~\cite{pro-ros-zoh:j:proportional-representation,bou-lu:c:chamberlin-courant,bet-sli-uhl:j:mon-cc}.
Nonetheless, these rules are so attractive that there is a growing
body of research on computing their winners exactly (e.g., through
integer linear programming
formulations~\cite{bra-pot:j:proportional-representation}, by means of
fixed-parameter tractability analysis~\cite{bet-sli-uhl:j:mon-cc}, 
by considering restricted
preference domains~\cite{bet-sli-uhl:j:mon-cc,yu-chan-elk:c:sptrees}) and
approximately~\cite{bou-lu:c:chamberlin-courant,sko-fal-sli:c:multiwinner,sko-fal-sli:c:monroe-cc-experiments}.
We continue this line of research by considering the complexity of
finding exact Chamberlin--Courant and Monroe winners for the case where
voters' preferences are single-crossing. Our results complement those
of Betzler et al.~\cite{bet-sli-uhl:j:mon-cc} for single-peaked
electorates.

Recall that voters are said to have single-crossing preferences if it is possible
to order them so that for every pair of candidates $a, b$ the voters
who prefer $a$ to $b$ form a consecutive block on one side of the
order and the voters who prefer $b$ to $a$ form a consecutive block on
the other side. For example, it is quite natural to assume that the
voters are aligned on the standard political left-right axis. Given
two candidates $a$ and $b$, where $a$ is viewed as more left-wing and
$b$ is viewed as more right-wing, the left-leaning voters would prefer $a$
to $b$ and the right-leaning voters would prefer $b$ to $a$.  
While real-life elections are typically too noisy to have this property, 
it is plausible that they may be close to single-crossing, 
and it is important to understand the complexity of the idealized model
before proceeding to study nearly single-crossing profiles 
(in the context of single-peaked elections this agenda has been successfully pursued
by Faliszewski et al.~\cite{fal-hem-hem:c:nearly-sp}).

Our main results are as follows: for single-crossing elections 
winner determination under Chamberlin--Courant's rule is in $\p$
(for every dissatisfaction function, and both for the utilitarian 
and for the egalitarian version of this rule), 
but under Monroe's rule it is $\np$-hard. Our hardness result for Monroe's rule
applies to the utilitarian setting with Borda dissatisfaction function.
Our algorithm for Chamberlin--Courant's rule extends to elections that have bounded 
{\em single-crossing width} (see~\cite{cor-gal-spa:c:sp-width,cor-gal-spa:c:spsc-width}).
Our proof proceeds by showing that for single-crossing elections
Chamberlin--Courant's rule admits an optimal assignment that
has the {\em contiguous blocks property}: the set of voters assigned to an elected representative
forms a contiguous block in the voters' order witnessing that the election is single-crossing.
This property can be interpreted as saying that each selected candidate represents
a group of voters who are fairly similar to each other, and we believe it to be desirable
in the context of proportional representation.

The $\np$-hardness result for Monroe's rule motivates us to search
for further domain restrictions that may make this problem tractable.
To this end, we focus on the egalitarian version of Monroe's rule 
and, following the example of Cornaz et al.~\cite{cor-gal-spa:c:sp-width}, consider
elections that, in addition to being single-crossing, are {\em narcissistic}, 
i.e., have the property that every candidate is ranked first by at least one voter.  
In parliamentary elections, narcissistic profiles are very natural: 
we expect all candidates to vote for themselves.  We provide a
polynomial-time algorithm for the egalitarian version of Monroe's rule
for all elections that belong to this class. Our algorithm is based on
the observation that for single-crossing narcissistic elections under
the egalitarian version of Monroe's rule there is always an optimal
assignment that satisfies the contiguous blocks property.  We show,
however, that this result does not extend to general single-crossing
elections or to the utilitarian version of Monroe's rule: in both
cases, requiring the contiguous blocks property may rule out all
optimal assignments.

In a sense, our result for single-crossing narcissistic elections is
not new: it can be shown that such elections are single-peaked (this
result is implicit in the work of Barber\`a and
Moreno~\cite{bar-mor:j:top-monotonicity}), and Betzler et
al.~\cite{bet-sli-uhl:j:mon-cc} provide a polynomial-time algorithm
for the egalitarian version of Monroe's rule for single-peaked
electorates. However, our algorithm 
has two significant advantages over the one of Betzler et al.: First,
it has considerably better worst-case running time, and second, it
produces assignments that have the contiguous blocks property.  In
contrast, if we formulate the analogue of the contiguous blocks
property for single-peaked elections, by considering the ordering of
the voters that is induced by the axis (see Section~\ref{sec:contig}
for details), we can construct an election where no optimal assignment
has the contiguous blocks property; this holds both for Monroe's rule
and for Chamberlin--Courant's rule (and both for the egalitarian
version and for the utilitarian version of either rule).

The paper is organized as follows. In Section~\ref{sec:prelim} we
provide the required background, give the definitions of Monroe's and
Chamberlin--Courant's rules, and define single-crossing and
single-peaked elections. Then, in Sections~\ref{sec:cc}
and~\ref{sec:m}, we discuss the complexity of winner determination
under Chamberlin--Courant's and Monroe's rules, respectively. We show
the limits of the contiguous blocks property approach in
Section~\ref{sec:contig}. We conclude the paper in
Section~\ref{sec:conclusions} by summarizing our results and
discussing future research directions.

\section{Preliminaries}\label{sec:prelim}
For every positive integer $s$, we let $[s]$ denote the set $\{1, \ldots, s\}$.
An \emph{election} is a pair $E = (C,V)$ where 
$C = \{c_1, \ldots, c_m\}$ is a set of candidates and $V = (v_1, \ldots, v_n)$ is an
ordered list of voters. Each voter $v \in V$ has a \emph{preference
  order} $\succ_v$, i.e., a linear order over $C$ that ranks all the
candidates from the most desirable one to the least desirable one.
For each voter $v \in V$ and each candidate $c \in C$, we denote by $\pos_v(c)$
the position of $c$ in $v$'s preference order (the top
candidate has position $1$ and the last candidate has position
$\|C\|$).  We refer to the list $V$ as the \emph{preference profile}.

Given an election $E=(C, V)$ and a subset of candidates $D\subset C$,
we denote by $V|_D$ the profile obtained by restricting the preference
order of each voter in $V$ to $D$.  We denote the concatenation of two
voter lists $U$ and $V$ by $U+V$; if $U$ consists of a single vote $u$
we simply write $u+V$.  A list $U$ is said to be a {\em sublist} of a
list $V$ (denoted by $U\subseteq V$) if $U$ can be obtained from $V$
by deleting voters.  An election $(C', V')$ is said to be a {\em
  subelection} of an election $(C,V)$ if $C'\subseteq C$ and
$V'=U|_{C'}$ for some $U\subseteq V$.  Given a subset of candidates
$A$, we denote by $\ordset{A}$ a fixed ordering of candidates in $A$
and by $\revset{A}$ the reverse of this ordering.  Given two disjoint
sets $A, B\subset C$, we write $\dots\succ A\succ B\succ\dots$ to
denote a vote where all candidates in $A$ are ranked above all
candidates in $B$.
\subsection{Chamberlin--Courant's and Monroe's Rules}
Both Chamberlin--Courant's rule and Monroe's rule rely on the notion
of a {\em dissatisfaction function} (also known as a {\em
  misrepresentation function}). This function specifies, for each
$i\in[m]$, a voter's dissatisfaction from being represented by
candidate she ranks in position~$i$.

\begin{definition}
For an $m$-candidate election, a {\em dissatisfaction function} 
is a nondecreasing function
$\alpha \colon [m] \rightarrow \mathbb{N}$ with $\alpha(1) = 0$.
\end{definition}
We will typically be interested in
families of dissatisfaction functions, $(\alpha^m)_{m=1}^{\infty}$,
with one function for each possible number of candidates.  In
particular, we will be interested in Borda dissatisfaction function
$\alpha_{\borda}^m(i) = \alphaborda(i) = i-1$. We assume that our
dissatisfaction functions are computable in polynomial time with
respect to $m$.

Let $k$ be a positive integer.  
A {\em $k$-CC-assignment function} for an election $E=(C, V)$ is a mapping 
$\Phi \colon V \rightarrow C$ such that $\| \Phi(V) \| \le k$.  A
{\em $k$-Monroe-assignment function} for $E$ is a $k$-CC-assignment function
that additionally satisfies the following constraints: $\| \Phi(V) \|
= k$, and for each $c \in C$ either $\|\Phi^{-1}(c)\| = 0$ or $\lfloor
\frac{n}{k} \rfloor \leq \|\Phi^{-1}(c)\| \leq \lceil \frac{n}{k}
\rceil$. That is, both assignment functions select (up to)
$k$ candidates, and a $k$-Monroe-assignment function
additionally ensures that each selected candidate is assigned to
roughly the same number of voters. For a given assignment function
$\Phi$, we say that voter $v\in V$ is \emph{represented} (in the parliament)
by candidate $\Phi(v)$.
There are several ways to measure the quality of an
assignment function $\Phi$ with respect to a dissatisfaction function $\alpha$;
we use the following two:
\begin{enumerate}\fixlist
\item $\ell_1(\Phi) = \sum_{i=1, \dots, n} \alpha(\pos_{v_i}(\Phi(v_i)))$, and
\item $\ell_\infty(\Phi) = \max_{i=1, \dots, n}\alpha(\pos_{v_i}(\Phi(v_i)))$.
\end{enumerate}
Intuitively, $\ell_1(\Phi)$ takes the utilitarian view
of measuring the sum of voters' dissatisfactions, whereas $\ell_\infty$ takes the egalitarian
view of looking at the worst-off voter only.

We are now ready to define the voting rules that are the subject of this paper.
\begin{definition}
  For every family of dissatisfaction functions
  $\alpha = (\alpha^m)_{m=1}^{\infty}$, every $R\in\{$CC, Monroe$\}$, and
  every $\ell\in\{\ell_1, \ell_\infty\}$, 
  an {\em $\alpha\hbox{-}\ell\hbox{-}R$ voting rule} is a mapping that takes 
  an election $E = (C,V)$ and a positive integer $k$ with $k\le \|C\|$ as its input, and
  returns a $k$-$R$-assignment function $\Phi$ for $E$ that minimizes
  $\ell(\Phi)$ (if there are several optimal assignments, the rule is
  free to return any of them).
\end{definition}

Chamberlin and Courant~\cite{cha-cou:j:cc} and Monroe~\cite{mon:j:monroe} 
proposed the utilitarian variants of
their rules and focused on Borda dissatisfaction function (though
Monroe also considered so-called $k$-approval dissatisfaction
functions). Egalitarian variants of both rules have been recently introduced by Betzler et
al.~\cite{bet-sli-uhl:j:mon-cc}.

\subsection{Single-Crossing and Single-Peaked Profiles}\label{sec:sc-sp}
The notion of single-crossing preferences dates back
to the work of Mirrlees~\cite{mir:j:single-crossing}; we also point
the reader to the work of Saporiti and
Tohm\'e~\cite{sap-toh:j:single-crossing-strategy-proofness} for some
settings where single-crossing preferences are studied.  
Formally, such elections are defined as follows.
\begin{definition}\label{def:sc}
  An election $E = (C,V)$, where $C = \{c_1, \ldots, c_m\}$ is a set of candidates 
  and $V = (v_1, \ldots, v_n)$ is an ordered list of
  voters, is {\em single-crossing} (with respect to the given
  order of voters) if for each pair of candidates $a$, $b$
  such that $a \succ_{v_1} b$, there exists a value $t_{a,b} \in [n]$
  such that $\{ i \in [n] \mid a \succ_{v_i} b \} = [t_{a,b}].$
\end{definition}
That is, as we sweep through the list of voters from the first one
towards the last one,  the relative order of every pair of candidates
changes at most once.

Definition~\ref{def:sc} refers to the ordering of the voters provided by $V$. 
Alternatively, one could simply require existence of an
ordering of the voters that satisfies the single-crossing property.
The advantage of our approach is that it simplifies notation, yet does not affect the complexity of
the problems that we study: one can compute 
an order of the voters that makes an election single-crossing
(or decide that such an order does not exist) in
polynomial time~\cite{elk-fal-sli:c:decloning,bre-che-woe:j:single-crossing}.

We also consider single-peaked elections~\cite{bla:j:rationale-of-group-decision-making}.
\begin{definition}
  Let $\succ$ be a preference order over candidate set $C$
  and let $\lhd$ be an order over $C$.
  We say that $\succ$ is {\em single-peaked with respect to $\lhd$} if for every triple of candidates
  $a, b, c \in C$ it holds that
  $
      ((a \lhd b \lhd c) \lor
       (c \lhd b \lhd a)) \implies (a \succ b \implies b \succ c).
  $
  An election $E = (C,V)$ is {\em single-peaked with respect to an
  order $\lhd$ over $C$} if the preference order of every voter $v\in V$ is
  single-peaked with respect to $\lhd$. An election $E = (C,V)$
  is {\em single-peaked} if there exists an order $\lhd$ over $C$ with respect to
  which it is single-peaked.
\end{definition}
If an election $E$ is single-peaked with respect to some order $\lhd$ then we
call $\lhd$ a \emph{societal axis} for $E$. There are polynomial-time algorithms
that given an election $E$ decide if it is single-peaked and if so, compute a
societal axis for it~\cite{bar-tri:j:stable-matching-from-psychological-model,esc-lan-ozt:c:single-peaked-consistency}.
Thus, just as in the case of single-crossing elections, we can freely assume that if an
election is single-peaked then we are given a societal axis as well.
\smallskip

\section{Chamberlin--Courant's Rule}\label{sec:cc}

We start our discussion by considering the complexity of
winner-determination under Cham\-berlin--Courant's rule, for the case of
single-crossing profiles.

\subsection{Single-Crossing Profiles}

A key observation in our analysis of Chamberlin--Courant's rule
is that for single-crossing
profiles there always exists an optimal $k$-CC-assignment function
where the voters
matched to a given candidate form contiguous blocks within the voters' order.
In what follows, we will say that assignments of this form
have the {\em contiguous blocks property}.
We believe that this property
is desirable from the social choice perspective: it means that voters
who are represented by the same candidate are quite similar, which makes
it easier for the candidate to act in a way that reflects
the preferences of the group he represents. Later, 
we will see that the contiguous blocks property also has useful algorithmic implications.

\begin{lemma}\label{lem:cc-blocks}
  Let $E = (C,V)$ be a single-crossing election, 
  where $C = \{c_1, \ldots, c_m\}$, $V = (v_1, \ldots,  v_n)$, and $v_1$
  has preference order $c_1 \succ \cdots \succ c_m$.
  Then for every $k\in [m]$, every
  dissatisfaction function $\alpha$ for $m$
  candidates, and every $\ell\in\{\ell_1, \ell_\infty\}$, there is an optimal
  $k$-CC assignment $\Phi$ for $E$ under $\alpha$-$\ell$-CC such that for each
  candidate $c_i \in C$, if $\Phi^{-1}(c_i) \neq \emptyset$ then there are
  two integers, $t_i$ and $t_i'$, $t_i \leq t_i'$, such that $\Phi^{-1}(c_i) =
  \{v_{t_i}, v_{t_i+1}, \ldots, v_{t_i'}\}$. Moreover, for each $i < j$ 
  such that $\Phi^{-1}(c_i) \neq \emptyset$ and $\Phi^{-1}(c_j) \neq \emptyset$ it
  holds that $t'_i < t_j$.
\end{lemma}
\begin{proof}
  Fix a single-crossing election $E = (C,V)$ with 
  $C = \{c_1, \ldots, c_m\}$ and $V = (v_1, \ldots, v_n)$, and let $\Phi$
  be an optimal $k$-CC-assignment function for $E$ under
  $\alpha$-$\ell$-CC. We assume without loss of generality that for each voter
  $v_i$ in $V$, the candidate $\Phi(v_i)$ is $v_i$'s most preferred
  candidate is $\Phi(V)$. Let $c_j$ be $v_1$'s least preferred candidate in $\Phi(V)$.
  Now consider some voter $v_i$ such that $\Phi(v_i) = c_j$. 
  We have $\Phi(v_{i'}) = c_j$ for every voter
  $v_{i'}$ such that $i' > i$. Indeed, 
  suppose for the sake of contradiction that
  $\Phi(v_{i'}) = c_k$ for $k \neq j$.  By our choice of $c_j$ we have
  $c_k \succ_1 c_j$. On the other hand, we have $c_j \succ_i c_k$ and
  $c_k \succ_{i'} c_j$, a contradiction with $E$ being a single-crossing election. Hence, the
  voters that are matched to $c_j$ by $\Phi$ form a consecutive block at the
  end of the preference profile.

  To see that for each $c \in \Phi(V)$ it holds that voters in
  $\Phi^{-1}(c)$ form a consecutive block, it suffices to delete
  $c_j$ and the voters that are matched to $c_j$ from the profile, decrease $k$ by
  one, and repeat the same argument.
\end{proof}

Lemma~\ref{lem:cc-blocks} suggests a dynamic programming
algorithm for Chamberlin--Courant's rule. 

\begin{theorem}\label{thm:cc-easy}
  For every family $\alpha$ of polynomial-time computable
  dissatisfaction functions and for $\ell \in
  \{\ell_1,\ell_\infty\}$, there is a polynomial-time algorithm that
  given a single-crossing election $E$ and a positive integer $k$
  finds an optimal $k$-CC assignment for $E$ under $\alpha$-$\ell$-CC.
\end{theorem}
\newcommand{\bigA}[3]{{{A[#1,#2,#3]}}}
\begin{proof}
  Let $E = (C,V)$ be our input
  single-crossing election, where $C = \{c_1, \ldots, c_m\}$, $V =
  (v_1, \ldots, v_n)$ and $v_1$ has preference order $c_1 \succ \cdots \succ c_m$,
  and let $k$ be the target parliament size.

  For every $i\in \{0\}\cup[n]$, $j\in[m]$, and $t\in[k]$
  we define $\bigA{i}{j}{t}$ to be the optimal
  $\ell$-aggregated dissatisfaction that can be achieved with
  a $t$-CC-assignment function when
  considering subelection $(C_j, V_i$), where $C_j = \{c_1, \ldots, c_j\}$
  and $V_i = (v_1, \ldots, v_i)$
  (clearly, $(C_j,V_i)$ is single-crossing).  It
  is easy to see that for every $i\in [n]$, $j\in[m]\setminus\{1\}$ and $k\in[t]\setminus\{1\}$
  the following recursive
  relation holds (in the equation below, we abuse notation and treat
  $\ell$ as the respective norm on real vectors, i.e., we assume
  that it maps a list of values to their sum
  (when $\ell = \ell_1$) or their maximum (when $\ell = \ell_\infty$)):
  \begin{align*}
  \bigA{i}{j}{t}= \min & \left\{\bigA{i}{j-1}{t}, \min_{i^{\ast}<i} \ell\left(\bigA{i^{\ast}}{j-1}{t-1}, %
\alpha(\pos_{v_{i^{\ast}+1}}(c_j)), %
\ldots, 
\alpha(\pos_{v_{i}}(c_j))\right)  %
\right\}.
  \end{align*}
  The idea of this recursive relation is to guess the
  first voter $v_{i^{\ast}+1}$ to be represented by $c_j$; the optimal
  representation of the preceding voters is found recursively, for
  assembly size $t-1$. To take care of the possibility that $c_j$ does
  not participate in the solution, we also take $\bigA{i}{j-1}{t}$ into account.

  The base cases of the above recursion are as follows.
  For every $j \in [m]$ and $t \in [k]$, we have $\bigA{0}{j}{t} = 0$.
  For every $i \in [n]$ and $j \in [m]$, it holds that
  \begin{align*}
  \bigA{i}{j}{1} =\min_{j'\leq j} \ell\left(\alpha(\pos_{v_1}(c_{j'})), \ldots, \alpha(\pos_{v_i}(c_{j'}))\right) \textrm{.}
  \end{align*}
  For every $i\in[n]$, $j \in [m]$, and $t \geq j$, 
  we have $\bigA{i}{j}{t} = 0$ (we match each voter to her
  top candidate). These %
  conditions suffice for our
  recursion to be well-defined. Using dynamic programming,
  we can compute in polynomial time (in fact, in time
  $O(mn^2k)$) the optimal dissatisfaction of the voters and a
  parliament that achieves it.
\end{proof}

\subsection{Extension to Profiles with Bounded Single-Crossing Width}

Following the ideas of Cornaz, Galand, and
Spanjaard~\cite{cor-gal-spa:c:sp-width,cor-gal-spa:c:spsc-width}, we
can extend our algorithm for Chamberlin--Courant's rule to profiles with so-called {\em bounded single-crossing
width}.

\begin{definition}
  A set $D$, $D \subseteq C$, is a {\em clone set} in an election $E=(C,V)$ 
  if each voter in $V$ ranks the candidates from $D$ consecutively (but not necessarily in the same order).
\end{definition}
\begin{definition}
  We say that an election $E=(C, V)$ has {\em single-crossing width} 
  (respectively, {\em single-peaked width}) at most $w$ if there exists a partition of $C$ into sets
  $D_1, \ldots, D_t$ such that (a) for each $i\in[t]$ the set $D_i$ is a
  clone set in $E$ and $\|D_i\| \leq w$, and (b) if we
  contract each $D_i$ in each vote to a single candidate $d_i$, then
  the resulting preference profile is single-crossing
  (respectively, single-peaked).
\end{definition}
Profiles with small single-crossing %
width may arise, e.g., in parliamentary elections where the
candidates are divided into (small) parties and the voters have single-crossing
preferences over the parties, but not necessarily over the
candidates.
Using the same techniques as Cornaz et al., we obtain the following
result.
\begin{proposition}\label{pro:cc-sc-w}
  For every family $\alpha$ of polynomial-time computable
  dissatisfaction functions and for every $\ell \in \{\ell_1,\ell_\infty\}$, 
  there is an algorithm that
  given an election $E=(C, V)$ 
  with $C=\{c_1, \dots, c_m\}$, $V=(v_1, \dots, v_n)$
  whose single-crossing width
  is bounded by $w$, a partition
  of $C$ into clone sets that witnesses this width bound,
  and a positive integer $k$, 
  finds an optimal $k$-CC assignment for $E$ under $\alpha$-$\ell$-CC,
  and runs in time $\poly(m, n, k, 2^w)$.
\end{proposition}
\newcommand{\tp}{{{{\mathrm{top}}}}}
\begin{proofsketch}
  Let $E = (C,V)$ be our input election, 
  and let $D_1, \ldots, D_s$ be a partition of $C$
  witnessing that the single-crossing width of $E$ is at most $w$; 
  assume that the order of the
  sets $D_1, \ldots, D_s$ is such that the preference order of the first voter in $V$
  is of the form $D_1 \succ D_2 \succ \dots \succ D_s$.
  We first observe that Lemma~\ref{lem:cc-blocks} generalizes easily
  to elections with a given partition into clone sets.
  Specifically, there exists an optimal
  $k$-CC assignment $\Phi$ for $E$ under $\alpha$-$\ell$-CC, where for
  each clone set $D_i$, if $\Phi^{-1}(D_i) \neq \emptyset$ (that is,
  if at least one candidate from $D_i$ is assigned to some voter)
  then:
   (a) there are two integers, $t_i$ and $t_i'$, $t_i \leq
    t_i'$, such that $\Phi^{-1}(D_i) = \{v_{t_i}, v_{t_i+1}, \ldots,
    v_{t_i'}\}$, and
   (b) for each $i < j$ such that $\Phi^{-1}(D_i) \neq
    \emptyset$ and $\Phi^{-1}(D_j) \neq \emptyset$ it holds that $t'_i
    < t_j$.
  That is, the voters matched to the candidates from a given clone set
  form a consecutive block within the voter order.

  Now it is easy to modify our dynamic programming algorithm for
  profiles with bounded single-crossing width. We guess an integer
  $j\in[s]$, a subset $D'_j$ of $D_j$ (the candidates
  from $D_j$ to join the assembly), and a voter $v_i\neq v_n$ such that 
  voters $v_{i+1}, \dots, v_n$ are represented by
  the candidates from $D'_j$. Note that assigning the candidates from
  $D'_j$ to these voters optimally is easy under Chamberlin--Courant's
  rule: each voter gets her most preferred candidate from
  $D'_j$. An optimal representation for $v_1, \dots, v_i$ is
  found recursively (for an appropriately smaller assembly).
  To implement guessing, we try all possible choices of
  $j$, all possible subsets of $D_j$, and all possible choices of $v_i$,
  and we use dynamic programming to implement the recursive
  calls efficiently, just as in the perfectly single-crossing case. 
  Since there are only $s2^w\|V\|$ possibilities to consider at each guessing step
  and $s \leq \|C\|$, we obtain the desired bound on the running time.
\end{proofsketch}

Naturally, for this result to be useful, we need an efficient
algorithm that computes single-crossing width of a profile and an
appropriate division into clone sets. Fortunately, such an algorithm
is provided by Cornaz et
al.~\cite{cor-gal-spa:c:spsc-width}. (Interestingly, a very similar
problem of finding a division into clones that results in a
single-crossing election with as many candidates as possible is
$\np$-hard~\cite{elk-fal-sli:c:decloning}).  As a consequence, we have
the following corollary (see the
books~\cite{nie:b:invitation-fpt,dow-fel:b:parameterized} for an
introduction to fixed-parameter complexity theory).

\begin{corollary}
  For every family $\alpha$ of polynomial-time computable
  dissatisfaction functions and for every $\ell \in
  \{\ell_1,\ell_\infty\}$, the problem of winner determination for
  $\alpha$-$\ell$-CC is fixed-parameter tractable with respect to
  the single-crossing width of the input profile.
\end{corollary}

\section{Monroe's Rule}\label{sec:m}

The results of Betzler et al.~\cite{bet-sli-uhl:j:mon-cc} suggest
that winner determination under Monroe's rule tends to be harder
than winner determination under Chamberlin--Courant's rule.
In this section, we show that this is also the case 
for single-crossing profiles: we prove 
that for the utilitarian variant of Monroe's rule with Borda dissatisfaction
function (perhaps the most natural variant of Monroe's rule) computing
winners is $\np$-hard, even for single-crossing elections.
We then complement this hardness result by showing that for the egalitarian version
of Monroe's rule winner determination is easy if we additionally
assume that the preferences are narcissistic.

\subsection{Hardness for General Single-Crossing Profiles}

This section is devoted to proving that winner determination under
Monroe's rule is NP-hard. The main idea of the proof is to reduce the
problem of winner determination for unrestricted profiles to the case
of single-crossing profiles.

\begin{theorem}\label{thm:monroe-sc-np}
  Finding a set of winners under $\alphaborda$-$\ell_{1}$-Monroe voting rule
  is $\np$-hard, even for single-crossing elections.
\end{theorem}
The proof of this theorem is somewhat involved. 
We first need the following two lemmas.

\begin{lemma}
  Consider an election $E = (C,V)$ with $C = \{c_1, \ldots, c_m\}$, 
  $V = (v_1, \ldots, v_n)$.  Let $A$ and $B$ be two disjoint sets of
  candidates such that $\|A\| = \|B\| = mn$.  For each $c_i \in C$,
  there is a single-crossing election $\mathrm{Adj}_V(A, c_i, B)$ with
  candidate set $A \cup B \cup \{c_i\}$ and voter list $V' = (v'_1,
  \ldots, v'_n)$ such that
$pos_{v_{j}'}(c_i) = mn + pos_{v_{j}}(c_i)$ for each $j\in [n]$, and
the profile $(v'_0, v'_1, \ldots, v'_n, v'_{n+1})$, where $v'_0$
    has preference order $a_1 \succ \dots \succ a_{\|A\|} \succ c_i
    \succ b_1 \succ \dots \succ b_{\|B\|}$ and $v'_{n+1}$ has
    preference order $b_1 \succ \dots \succ b_{\|B\|} \succ c_i \succ
    a_1 \succ \dots \succ a_{\|A\|}$, is also single-crossing.
\end{lemma}

\begin{proof}
Set $A = \{a_1, \dots, a_{mn}\}$ and $B = \{b_1, \dots, b_{mn}\}$. 
Fix a candidate $c_i \in C$. We build the election
$\mathrm{Adj}_V(A, c_i, B)$ as follows.
We set $v'_1$'s preference order to be
\begin{align*}
a_1 \succ a_2  \succ \dots &\succ a_{mn} \succ b_1 \succ b_2 \succ \dots \\&\succ b_{pos_{v_1}(i) - 1} \succ  c_i \succ  b_{pos_{v_1}(i)} \succ \dots \succ b_{mn}.
\end{align*}
For $1\le j\le n-1$, we build the preference order of voter $v_{j+1}$ based on the
preference order already constructed for $v_j$. Given that the preference order of $v_j$ is of the form
\begin{align*}
a_1 \succ a_2 \succ \dots &\succ a_{x} \succ b_1 \succ b_2 \succ \dots \succ  b_{y} \succ c_i \\&\succ  b_{y+1} \succ \dots \succ b_{mn} \succ a_{x+1} \succ \dots \succ a_{mn},
\end{align*}
we construct the preference order of $v_{j+1}$ either by moving some of the candidates from $B$ to precede $c_i$ or
by moving some of the candidate from $A$ to follow $c_i$.
Specifically, we do the following.
First, we compute $d = pos_{v_{j+1}}(i) - pos_{v_{j}}(i)$; note that $-m < d < m$. If $d \geq 0$ then
we set $v_{j+1}$'s preference order to be
\begin{align*}
a_1 \succ a_2 \succ \dots &\succ a_{x} \succ b_1 \succ b_2 \succ \dots \succ  b_{y+d} \succ c_i \\&\succ  b_{y+d+1} \succ \dots \succ b_{mn} \succ a_{x+1} \succ \dots \succ a_{mn}.
\end{align*}
If $d < 0$ then we set $v_{j+1}$'s preference order to be
\begin{align*}
a_1 \succ a_2 \succ \dots &\succ a_{x-d} \succ b_1 \succ b_2 \succ \dots \succ  b_{y} \succ c_i \\&\succ  b_{y+1} \succ \dots \succ b_{mn} \succ a_{x-d} \succ \dots \succ a_{mn}.
\end{align*}
If $d=0$ then we set $v_{j+1}$'s preference order to be the same as $v_j$'s.
Clearly, to construct each vote it suffices to shift forward or backward a block of at most $m$ candidates
and since both $A$ and $B$ contain $mn$ candidates, doing so is always possible. Finally, it is clear that
we never change the relative order of the candidates within $A$ and within $B$, and that the resulting
profile is single-crossing, even if we prepend $v'_0$ and append $v'_{n+1}$ to it.
\end{proof}

\begin{lemma}\label{lem:all-m}
  For every pair of positive integers $k, n$ such that
  $k$ divides $n$, and every set $C = \{c_1, \ldots, c_m\}$ of
  candidates, there is a single-crossing profile $\rot(C)$ with
  $(\frac{n}{k} + 1)m$ voters such that each candidate $c_{i} \in C$
  is ranked first by exactly $(\frac{n}{k} + 1)$ voters.
\end{lemma}
\begin{proof}
We build a list $V = V_1 + \cdots + V_m$ of voters, where each $V_i$, $i\in[m]$,
contains $\frac{n}{k} + 1$ voters. The preference order of each voter in $V_i$ is
$
  c_i \succ c_{i+1} \succ \cdots c_{m} \succ c_{i-1} \succ \cdots \succ c_2 \succ c_1.
$
It is clear that a thus-constructed profile is single-crossing.
Note that the preference order of the last voter in this profile is the
reverse of the preference order of the first voter.
\end{proof}

We extend the notation introduced in Lemma~\ref{lem:all-m} to apply to orders
of candidates.  That is, if $\ordset{C}$ is an order of candidates in
$C$, then by $\rot(\ordset{C})$ we denote the election that we would
construct in Lemma~\ref{lem:all-m} if the first voter's preference
order was $\ordset{C}$ (i.e., if we took $c_1$ to be the top
candidate according to $\ordset{C}$, $c_2$ to be the second one, and
so on).  By $\rot^{-1}(\ordset{C})$ we denote an order $\ordset{C'}$ of
the candidates in $C$ such that $\rot(\ordset{C'})$ produces an
election where the last voter has preference order $\ordset{C}$. 
With these lemmas and notation available, we are ready to give our
proof of Theorem~\ref{thm:monroe-sc-np}.\medskip

\newcommand{\msucc}{\ }

\begin{table}[t!] 
\centering
  \setlength{\tabcolsep}{0pt}
  \begin{tabular}{lll}
  $V_{1}:$&$ H \msucc \rot^{-1}(F_1 \msucc \dots \msucc F_m \msucc E \msucc E_m \msucc \dots \msucc E_{1})\msucc c_1 \msucc \dots \msucc c_m \msucc  D_1 \msucc \dots \msucc D_m \msucc G_1 \dots G_m \msucc G$ \\
  $V_{2}:$&$ H \msucc \rot(\rot^{-1}(F_1 \msucc \dots \msucc F_m \msucc E \msucc E_m \msucc \dots \msucc E_{1})) \msucc c_1 \msucc \dots \msucc c_m \msucc D_1 \msucc \dots \msucc D_m \msucc G_1 \dots G_m \msucc G $\\ \vspace{-1.5mm} \\
  $v_3^1:$&$ H \msucc  F_1  \cdots  F_m \msucc E \msucc E_m \msucc \dots \msucc E_{2}  \msucc c_1 \msucc E_1 \msucc c_2  \msucc \dots \msucc c_m \msucc  D_1 \msucc \dots \msucc D_m \msucc G_1 \msucc \dots \msucc G_m \msucc G$ \\
  $v_{3}^{2}:$&$ H \msucc F_1  \cdots  F_m \msucc E \msucc D_1  \msucc E_m \msucc \dots \msucc E_{3} \msucc c_2 \msucc E_{2}  \msucc c_3 \msucc \dots \msucc c_m \msucc c_1 \msucc E_1 \msucc D_2 \msucc \dots \msucc  D_m \msucc  G_1 \msucc \dots \msucc G_m \msucc G $\\
  & \quad $\vdots$ \\
  $v_{3}^{m}:$&$ H \msucc F_1  \cdots  F_m \msucc E \msucc D_1 \msucc \dots \msucc D_{m-1} \msucc c_m  \msucc E_{m} \msucc c_{m-1} \msucc \dots \msucc c_1 \msucc E_{m-1} \msucc \dots E_{1}  \msucc D_m \msucc G_1 \msucc \dots \msucc G_m \msucc G $\\ \vspace{-1.5mm} \\
  $V_{4}:$&$ H \msucc D_1 \msucc \dots \msucc  D_m \msucc \mathrm{Adj}(F_1, c_m, G_1) \msucc \dots   \msucc \mathrm{Adj}(F_m, c_1, G_m)  \msucc E \msucc E_m \msucc \dots \msucc E_{1} \msucc G$ \\
  $V_{5}:$&$ H \msucc \rot(D_1 \msucc \dots \msucc  D_m \msucc  G_1 \msucc \dots \msucc G_m \msucc  G) \msucc c_m \msucc \dots \msucc c_1 \msucc F_1 \msucc \dots \msucc F_{m} \msucc E \msucc E_m \msucc \dots \msucc E_{1}$
  \end{tabular}
  \caption{\label{tab:monroe-sc-np}The profile used in the proof of Theorem~\ref{thm:monroe-sc-np}.
  For each voter list $V_i$, $1 \leq i \leq 5$, and for each voter $v$ in $V_i$ we list the (sets of) candidates in the
  order of $v$'s preference (we omit the ``$\succ$'' symbol for readability).
    Whenever we list a set of candidates as a part of an order, we assume that the candidates in this set are ordered in some
  fixed, easily-computable way (for candidates in $H$ we fix this order to be $h_1 \succ \dots \succ h_{m-k}$).
  Further, when in a line describing a preference order of an entire collection of voters $V_r = (v_1, \ldots, v_s)$
  (specifically, for us $r$ is either $2$, $4$, or $5$, and $s$ is the number of voters in this list of voters) 
  we include a profile $V' = (v'_1, \ldots, v'_s)$ (in our
  case $V'$ is either an $\rot$-profile or an $\mathrm{Adj}$-profile), then
  we mean that for each voter $v_i$, $i\in[s]$, in $V_r$, 
  this part of this voter's preference order is the preference order of $v'_i$ in $V'$.
  }
  \end{table}

\begin{proofof}{Proof of Theorem~\ref{thm:monroe-sc-np}}
  Let $I$ be an instance of the problem of finding $k$ winners under
  $\alphaborda$-$\ell_{1}$-Monroe rule, and let $(C,V)$ be the election
  considered in $I$.  Set $n = \|V\|$ and $m = \|C\|$. We assume that
  $n$ is divisible by $k$ and that $n > k$ (computing
  $\alphaborda$-$\ell_1$-Monroe winners is still $\np$-hard under
  these
  assumptions~\cite{bet-sli-uhl:j:mon-cc,sko-fal-sli:c:multiwinner}).
  We will show how to construct in polynomial time an instance
  $I_{sc}$ of the problem of finding winners under
  $\alphaborda$-$\ell_{1}$-Monroe where the election is
  single-crossing so that it is easy to extract the set of
  winners for $I$ from the set of winners for $I_{sc}$.

  We construct $I_{sc}$ in the following way.  First, we define the
  candidate set $C_{sc}$ to be the union of the following disjoint
  sets (we provide names of the candidates only where relevant
  and abbreviate $\sum_{i=1}^m$ to $\sum_i$):
  \begin{enumerate}\fixlist
  \item 
  $H = \{h_1, \ldots, h_{m-k}\}$, where $\|H\| = m-k$;
  \item 
  $F_1, \dots F_m$, where $\|F_i\| = mn$ for each $i\in[m]$;
  \item 
  $E_1, \dots E_m$, where $\|E_i\| = 2m^2n + m + (m-i)(2mn+1)\frac{n}{k}$ for each $i\in [m]$;
  \item $E$, where $\|E\| = m^2n + m$;
  \item $D_1, \dots, D_m$, where $\|D_i\| = \|E_i\|$ for each $i\in [m]$;
  \item $G_1, \dots, G_m$, where $\|G_i\| = \|F_i\| = mn$ for each $i\in [m]$;
  \item $G$, where $\|G\| = (\sum_i \|F_i\| + \|E\|)$;
  \item $C' = C = \{c_1, \dots, c_m\}$.
  \end{enumerate}
  The ordered list $V_{sc}$ of voters consists of the following five
  sublists (we only give names to those voters to whom we will refer
  directly later; whenever sufficient, we only give the number of
  voters in a given list):
  \begin{enumerate}\fixlist
  \item $V_1$, $\|V_1\| = \|H\|\frac{n}{k} = (m-k)\frac{n}{k}$;
  \item $V_2$, $\|V_2\| = (\sum_{i}\|F_i\| + \sum_{i}\|E_i\| +
    \|E\|)(\frac{n}{k} + 1)$;
  \item $V_3 = (v_{3}^1, \ldots, v_{3}^m)$, $\|V_3\| = m$;
  \item $V_4$, $\|V_4\| = n$;
  \item $V_5$, $\|V_5\| = (\sum_{i}\|D_i\| + \sum_{i}\|G_i\| +
    \|G\|)(\frac{n}{k} + 1)$.
  \end{enumerate}
  We give the preferences of the voters in
  Table~\ref{tab:monroe-sc-np}.  In the thus-defined profile our goal
  is to find a parliament of size $k_{sc} = \|C_{sc}\| - (m-k)$.
  Consequently, each selected candidate should be assigned to
  $\frac{n}{k} + 1$ voters.

  We claim that optimal solutions $\Phi_{sc}$ for $I_{sc}$ satisfy the
  following conditions:
  \begin{enumerate}\fixlist
  \item[(i)] Each candidate $c \in F_1 \cup \dots \cup F_m \cup E \cup
    E_m \cup \dots \cup E_{1}$ is a winner and is assigned to those
    voters from $V_1+V_2$ that rank $c$ in position $\|H\| + 1=
    m-k+1$ (note that only one of these candidates can
    be assigned to (some of the) voters in $V_1$).
  \item[(ii)] Each candidate $c \in D_1 \cup\dots \cup D_m \cup G_1
    \cup\dots\cup G_m \cup G$ is a winner and is assigned to those
    voters from $V_5$ that rank $c$ in position $\|H\| + 1 = m-k+1$.
  \item[(iii)] Each candidate $h_{i} \in H$ is a winner and is
    assigned to $\frac{n}{k}+1$ voters from $V_1+V_2+V_3$ (exactly
    $\|H\|$ voters from $V_3$ have some candidate from $H$ assigned to
    them); each such voter ranks $h_{i}$ in position $i$.
  \item[(iv)] Exactly $k$ candidates from $C'$ are winners. Each of
    them is assigned to $\frac{n}{k}$ voters in $V_4$ and to one voter
    in $V_3$ that ranks him highest.
  \item[(v)] The $k$ winners from $C'$ (let us call them $w_1, \ldots,
    w_k$) are also $\alphaborda$-$\ell_1$-Monroe winners in $I$ and
    each of them is assigned in $I_{sc}$ to the voters corresponding
    to those from the $I$-solution.
  \end{enumerate}

  Let us now show that indeed the optimal solution is of this form.
  First, we make the following observations:
  \begin{enumerate}\fixlist
  \item[(a)] By a simple counting argument, at least $k$ of the
    candidates from $C'$ must be included in the optimal solution.

  \item[(b)] For each candidate $h_i$ in $H$, if $h_i$ is part of the
    optimal solution then $h_i$ is ranked in the $i$-th position in
    the preference order of the voters to which $h_i$ is assigned
    (candidates from $H$ are always ranked first, in the order $h_1
    \succ \dots \succ h_m$).

  \item[(c)] For each candidate $c \in C_{sc} \setminus (C' \cup H)$, if $c$
    is included in the optimal solution then each voter to which $c$
    is assigned ranks $c$ in position $m-k+1$ or worse (this is
    because every voter's top $m-k$ positions are taken by the
    candidates from $H$).

  \item[(d)] Each voter in $V_1+V_2+V_5$ ranks each candidate in $C'$
    position worse than \[p_1 = \|H\| + \sum_i \|E_i\| + \sum_i
    \|F_i\| + \|E\| > \|H\| + \sum_i \|E_i\| + 2m^2n + m.\]

  \item[(e)] Each voter in $V_4$ ranks each candidate in $C'$ in
    position better than \[p_2 = \|H\| + \sum_i \|D_i\| + \sum_i
    \|F_i\| + \sum_i \|G_i\| + m = \|H\| + \sum_i \|E_i\| + 2m^2n +
    m,\] 
   but worse than $p_3 = \|H\| + \sum_i \|E_i\|$.

  \item[(f)] $p_1 > p_2$.

  \item[(g)] For each candidate $c \in C'$, there is exactly one voter
    in $V_3$ that ranks $c$ in a position no worse than
    \begin{align*}
      p_4 &= \|H\| + \sum_i\|F_i\| + \sum_{i<m}\|E_i\| + \|E\| + 1 \\
          &= \|H\| + \sum_{i<m}\|E_i\| + \|E\| + m^2n + 1 \\
          &< \|H\| + \sum_{i}\|E_i\| = p_3;
    \end{align*}
    all other voters in $V_3$ rank $c$ in a position worse
    than \[p_5 = \|H\| + \sum_i \|E_i\| + \sum_i \|F_i\| + \|E\| =
    \|H\| + \sum_i \|E_i\| + 2m^2n + m =p_2.\]
  \end{enumerate}

  Let $\Phi$ be an optimal assignment function among those that use
  exactly $k$ candidates from $C'$. We claim that $\Phi$ satisfies
  conditions (i)--(iv).
  This is so, because assigning voters from $V_4$ to candidates
  other than those in $C'$ will result in a strictly worse assignment (the
  assignment would get worse for the candidates in $C'$ because of
  points (d), (e), (f) and (g), and it would not improve for the other
  candidates because of points (b) and (c)). Similarly, each of the $k$
  selected candidates from $C'$ should be assigned to exactly one
  voter from $V_3$---the one that ranks this candidate highest. Once we assign
  the $k$ winners from $C'$ to the voters in $V_4$ and to $k$ voters
  in $V_3$, the optimal way to complete the assignment is to do so as
  described in conditions (i)--(iv).

  Let $\Phi$ be an optimal assignment function for $I_{sc}$ that uses
  exactly $k$ candidates from $C'$ and that satisfies conditions
  (i)--(iv).  We now prove that it also satisfies condition (v).
  Consider a candidate $c_{i} \in C'$ that is included in the set of
  winners under $\Phi$. Let $V^{c_i}_{4}$ be the subcollection of the
  voters from $V_4$ that are assigned to $c_i$ under $\Phi$
  (naturally, $\|V^{c_i}_4\| = \frac{n}{k}$).  Let $V^{c_i}$ be the
  subcollection of $V$ containing the voters corresponding to those
  in $V_4^{c_i}$ (again, $\|V^{c_i}\| = \frac{n}{k}$).  Let
  $s(V_{4}^{c_i})$ be the dissatisfaction of the voters in $V_4^{c_i}$
  under $\Phi$ and let $s(V^{c_i})$ denote the dissatisfaction the
  voters in $V^{c_i}$ would have if they were assigned to $c_i$ (in
  $I$). The total dissatisfaction of the voters assigned to $c_i$
  under $\Phi$ is:
  \begin{align*}
    & \textstyle(\|H\| + \sum_j \|E_j\| + \sum_j \|F_j\| + \|E\| - \|E_i\|) + s(V_{4}^{c_i}) = \\
    & \textstyle(\|H\| + \sum_j \|E_j\| + \sum_j \|F_j\| + \|E\| - 2m^2n - m - (m-i)(2mn+1)\frac{n}{k}) \\
    & \textstyle+ \frac{n}{k}(\|H\| + \sum_j \|D_j\| + (m-i)(2mn+1) + mn) + s(V^{c_i})) = \\
    & \textstyle(\frac{n}{k}+1)(\|H\| + \sum_j \|E_j\| + mn) +
    s(V^{c_i})\textrm{,}
  \end{align*}
  which shows that the dissatisfaction of the voters in $I_{sc}$ that
  are assigned to $c_i$ under $\Phi$ differs from the dissatisfaction
  of the respective voters in $I$, had they been assigned to $c_i$,
  only by a value that depends on $n$, $m$, and $k$ (but not on
  $i$). Thus condition (v) holds.

  It remains to show that an optimal assignment function for $I_{sc}$
  uses exactly $k$ candidates from $C'$.  Let $\Phi_{sc}$ be an
  optimal assignment function for $I_{sc}$ that satisfies conditions
  (i)--(v) (and thus uses exactly $k$ candidates from $C'$). Let
  $\Phi'$ be an assignment function for $I_{sc}$ that uses more thank
  $k$ candidates from $C'$. We will show that the total
  dissatisfaction under $\Phi'$ is higher than under $\Phi_{sc}$.  It
  is easy to see that the average dissatisfaction of the voters
  in $\Phi_{sc}^{-1}(C_{sc}\setminus C')$ under $\Phi_{sc}$ is lower or equal
  than that of the voters in $(\Phi')^{-1}(C_{sc}\setminus C')$ under $\Phi'$
  (this follows by contrasting properties (i)--(iii) and observations
  (b) and (c)).
  
  By the same reasoning as in the proof of property (v), we note that
  for each candidate $c_i \in C'$, the dissatisfaction of the voters
  that $\Phi'$ assigns to $c_i$ can be lower bounded by the
  dissatisfaction for the case where $c_i$ is assigned to the voter
  from $V_3$ that ranks $c_i$ highest and to $\frac{n}{k}$ voters from
  $V_4$ that rank $c_i$ highest. The voter from $V_3$ ranks $c_i$ in a
  position no better than $\|H\| + \sum_j \|E_j\| + \sum_j \|F_j\| +
  \|E\| - \|E_i\|$ and each of the $V_4$ voters ranks $c_i$ in
  position no better than $\|H\| + \sum_j \|D_j\| + (m-i)(2mn+1) +
  mn$.  Thus, under $\Phi'$, the dissatisfaction of the voters to
  whom $c_{i}$ is assigned can be lower bounded by:
  \begin{align*}
    & \textstyle \|H\| + \sum_j \|E_j\| + \sum_j \|F_j\| + \|E\| - \|E_i\| \\
    & \textstyle + \frac{n}{k}(\|H\| + \sum_j \|D_j\| + (m-i)(2mn+1) + mn) \\
    & \textstyle = \|H\| + \sum_j \|E_j\| + \sum_j \|F_j\| + \|E\|  \\
    & \textstyle + \frac{n}{k}(\|H\| + \sum_j \|D_j\| + mn) - 2m^2n - m.
  \end{align*}
  Similarly, if $c_i$ is selected under $\Phi_{sc}$ then the
  dissatisfaction of the voters to whom $c_i$ is assigned under
  $\Phi_{sc}$ can be upper bounded as follows (the idea of the upper
  bound is similar to the one for the lower bound above, except now
  we take the upper bound regarding the position of $c_i$ in the preference
  orders of voters from $V_4$):
  \begin{align*}
    & \textstyle \|H\| + \sum_j \|E_j\| + \sum_j \|F_j\| + \|E\| - \|E_i\| \\
    & \textstyle + \frac{n}{k}(\|H\| + \sum_j \|D_j\| + (m-i)(2mn+1) + mn + m) \\
    & \textstyle = \|H\| + \sum_j \|E_j\| + \sum_j \|F_j\| + \|E\| \\
    & \textstyle + \frac{n}{k}(\|H\| + \sum_j \|D_j\| + mn + m) - 2m^2n
    - m.
  \end{align*}
  Note that neither of the bounds depends on $i$ and that the
  difference between the upper bound for $\Phi_{sc}$ and the lower
  bound for $\Phi'$ is $m\frac{n}{k}$.  Thus for each subset of
  $k$ candidates from $C'$ that are assigned under $\Phi'$, the total
  dissatisfaction these candidates impose is, at best, better by an
  additive factor of $mn$ than the dissatisfaction imposed by the $k$
  candidates selected from $C'$ under $\Phi_{sc}$. However, under
  $\Phi'$ there are at least $\frac{n}{k}$ voters outside of $V_4$ 
  that are assigned to candidates in $C'$. Each of these 
  voters ranks the candidate assigned to her in a position 
  worse than $p_2 > 2m^2n + m$.  On the other hand, under
  $C_i$ each of these voters is assigned to some candidate that 
  she ranks in a position no worse than $m -k + 1$.  Since
  $\frac{n}{k}(2m^2n + m - m + k -1) > mn$, it holds that
  $\Phi_{sc}(V_{sc}) < \Phi'(V_{sc})$.

  We conclude that an optimal assignment function $\Phi_{sc}$ assigns
  voters to exactly $k$ candidates from $C'$ and that the
  dissatisfaction of the voters in $I_{sc}$ under $\Phi_{sc}$ is equal
  to the optimal dissatisfaction of the voters in $I$ plus an easily
  computable value that depends on $m$, $n$, and $k$ only. This
  completes the proof.
\end{proofof}

Betzler et al.~\cite{bet-sli-uhl:j:mon-cc} have shown a similar
hardness result for single-peaked elections; however, their construction 
uses an artificial dissatisfaction function rather than Borda. 
The complexity of winner
determination under $\alpha_\borda$-$\ell_1$-Monroe for single-peaked
elections is still an open question. As our result answers this question
in the case of single-crossing elections, it is
tempting to ask if our proof approach could be used for 
single-peaked elections. Unfortunately, this does not seem to be the case.
The difficulty lies in jointly implementing voters $V_3+V_4$ (and, in particular, 
positioning the candidates $c_1, \ldots, c_m$).

\subsection{$\boldsymbol{\ell_\infty}$-Monroe for Single-Crossing Narcissistic Profiles}
Given our hardness result for Monroe's rule,
it is natural to ask if we can further restrict the problem of computing Monroe's winners
to obtain tractability. To this end, we focus on the egalitarian version of Monroe'e rule
(the results of Betzler et al.~\cite{bet-sli-uhl:j:mon-cc} suggest that it is likely 
to be more tractable than the utilitarian version of this rule), 
and consider an additional domain restriction, namely, {\em narcissistic} preferences. 

An election is said to be {\em narcissistic} if every candidate is
ranked first by at least one voter.  Intuitively, such elections arise
when candidates are allowed to vote for themselves.  This notion was
introduced by Bartholdi and
Trick~\cite{bar-tri:j:stable-matching-from-psychological-model}, and
was used in the context of fully proportional representation by Cornaz
et al.~\cite{cor-gal-spa:c:sp-width}.  It turns out that it is useful
in our setting, too: we will show that the egalitarian version of
Monroe's rule admits an efficient winner determination algorithm under
single-crossing narcissistic preferences.

\begin{lemma}\label{lem:n-cons-contig}
  Let $E=(C, V)$ be a single-crossing narcissistic election with 
  $C = \{c_1, \ldots, c_m\}$, $V = (v_1, \ldots, v_n)$, where $v_1$ has
  preference order $c_1 \succ \cdots \succ c_m$. For every $k\in [m]$
  and every dissatisfaction function $\alpha$ for $m$ candidates,
  there is an optimal $k$-Monroe assignment $\Phi$ for $E$ under
  $\alpha$-$\ell_\infty$-Monroe  such that for each
  candidate $c_i \in C$, if $\Phi^{-1}(c_i) \neq \emptyset$ then there
  are two integers, $t_i$ and $t_i'$, $t_i \leq t_i'$, such that
  $\Phi^{-1}(c_i) = \{v_{t_i}, v_{t_i+1}, \ldots,
  v_{t_i'}\}$. Moreover, for each $i < j$ such that $\Phi^{-1}(c_i)
  \neq \emptyset$ and $\Phi^{-1}(c_j) \neq \emptyset$ it holds that
  $t'_i < t_j$.
\end{lemma}
\begin{proof}
  Suppose that we seek a parliament of size $k$.
  Let $\alpha$ be some dissatisfaction function, let $\Phi$ be an
  optimal $k$-Monroe-assignment function for $E$ under
  $\alpha$-$\ell_\infty$-Monroe, and let $t$ be the smallest value
  such that under $\Phi$ each voter is assigned to a candidate that
  this voter ranks in position $t$ or higher.  We will now show how to transform
  $\Phi$ into the form required in the statement of the lemma.

  Let $s = \max\{ s' \mid$ there is $v_i$ such that $\Phi(v_i) = c_{s'}
  \}$; that is, let $c_s$ be the last candidate in $v_1$'s preference
  order that is assigned to some voters. Let $n_s$ be the number of
  voters to whom $c_s$ is assigned.  We claim that there is a
  $k$-Monroe-assignment function $\Phi''$ with the same egalitarian
  utility as that of $\Phi$ that assigns $c_s$ to the voters $v_{n-n_s+1},
  \ldots, v_n$.

  If $\Phi$ itself satisfies this requirement, we are done. Otherwise, we
  transform it as follows. Pick the smallest $i$ 
  such that $\Phi(v_i) = c_s$. There must be some voter $v_j$,
  $j > i$, such that $\Phi(v_j) \neq c_s$ (otherwise $\Phi$ would have
  satisfied our requirement). Let $c_r = \Phi(v_j)$.  We set
  $$
  \Phi'(v_\ell)=
  \begin{cases}
  \Phi(v_\ell) &\text{if $\ell\neq i, j$}\\
  c_r          &\text{if $\ell=i$}\\
  c_s	     &\text{if $\ell=j$}.
  \end{cases}
  $$
  Clearly. $\Phi'$ is a
  $k$-Monroe-assignment function for $E$ and we now show that
  $\ell_\infty(\Phi') \leq \ell_\infty(\Phi)$. We consider the
  following two cases.
  \begin{description}

  \item[$\boldsymbol{v_i}$ prefers $\boldsymbol{c_s}$ to
    $\boldsymbol{c_r}$.]  Since $v_1$ prefers $c_r$ to $c_s$ and we
    assume that $v_i$ prefers $c_s$ to $c_r$, it must be the case that
    $v_j$ also prefers $c_s$ to $c_r$. Further, since $E$ is
    narcissistic, there is a voter $v_\ell$, $\ell < i$, who ranks $c_r$
    first. This means that $v_i$ ranks $c_r$ at least as highly
    as $v_j$ does. Indeed, otherwise there would be a candidate
    $c$ such that $v_i$ prefers $c$ to $c_r$ and $v_j$ prefers $c_r$
    to $c$; this is impossible since $v_\ell$ prefers $c_r$ to $c$ and
    $E$ is single-crossing.  This means that
    $\ell_\infty(\Phi') \leq \ell_\infty(\Phi)$.

  \item[$\boldsymbol{v_i}$ prefers $\boldsymbol{c_r}$ to
    $\boldsymbol{c_s}$.] This means that assigning $c_r$ to $v_i$ does
    not increase the dissatisfaction induced by the assignment. We
    have to show that assigning $c_s$ to $v_j$ does not increase it
    either. If $v_j$ prefers $c_s$ to $c_r$, then we are done.
    Thus, we can assume that $v_j$
    prefers $c_r$ to $c_s$. Since $V$ is narcissistic, there
    is a vote $v_q$ with $q > j$ where $c_s$ is ranked first. This
    means that $v_j$ ranks $c_s$ at least as highly as
    $v_i$ does. Indeed, otherwise there would be some candidate $c$ such that $v_i$
    prefers $c_s$ to $c$, $v_j$ prefers $c$ to $c_s$, and $v_q$
    prefers $c_s$ to $c$; this is impossible, since $E$ is single-crossing.
    Thus $\ell_\infty(\Phi') \leq \ell_\infty(\Phi)$.
\end{description}

This proves that $\Phi'$ is a $k$-Monroe assignment with egalitarian
dissatisfaction no worse than that of $\Phi$. By repeating the same
reasoning sufficiently many times, we eventually reach an assignment
function $\Phi''$ where $c_s$ is assigned exactly to the voters
$v_{n-n_s+1}, \ldots, v_n$: each iteration gets us closer to this goal. 
We then continue in the same way to handle the rest of the elected candidates. 
That is, we set $s^{(2)} = \max\{ s' \mid$ there is $v_i$ such
that $\Phi''(v_i) = c_{s'}$ and $s' < s\}$, we set $n_{s^{(2)}}$ to be
the number of voters to which $\Phi''$ assigns $c_{s^{(2)}}$, and
transform $\Phi''$ to be a $k$-Monroe-assignment function that assigns
$c_{s^{(2)}}$ to voters $v_{n-n_s-n_{s^{(2)}}+1}, \ldots, v_{n-n_s}$,
and $c_s$ to voters $v_{n-n_s+1}, \ldots, v_{n}$. We repeat in the
same manner until we reach a $k$-Monroe-assignment function that
satisfies the statement of the lemma.
\end{proof}

Based on Lemma~\ref{lem:n-cons-contig}, 
it is easy to construct a dynamic programming algorithm for
$\ell_\infty$-Monroe.

\begin{theorem}\label{thm:monroe-n-cons-easy}
  For every family $\alpha$ of polynomial-time computable
  dissatisfaction functions, there is a polynomial-time algorithm that
  given a single-crossing narcissistic election $E$ and a positive
  integer $k$ finds an optimal $k$-Monroe assignment for $E$ under
  $\alpha$-$\ell_\infty$-Monroe.
\end{theorem}
\begin{proof}
  Let $E = (C,V)$ be a single-crossing narcissistic election, let
  $\alpha$ be our dissatisfaction function, and let $k$ be the size of
  the parliament that we seek. We assume that $C = \{c_1, \ldots,
  c_m\}$, $V = (v_1, \ldots, v_n)$, and that $v_1$'s preference order
  is $c_1 \succ \cdots \succ c_m$. Based on
  Lemma~\ref{lem:n-cons-contig}, we give a dynamic programming
  algorithm for $\alpha$-$\ell_\infty$-Monroe.

  For each $i,j,t$, $0 \leq i \leq n$, $0 \leq j \leq m$, $0 \leq t
  \leq k$, we define $\bigA{i}{j}{t}$ to be the lowest dissatisfaction
  that one can achieve by assigning to the voters $v_1, \ldots, v_i$
  exactly $t$ candidates from the set $C_j = \{c_1, \ldots, c_j\}$ in
  such a way that each candidate is assigned to either $\lceil
  \frac{n}{k} \rceil$ consecutive voters or to $\lfloor \frac{n}{k} \rfloor$ consecutive
  voters; we set $\bigA{i}{j}{t} = \infty$ if such an assignment does
  not exist. In particular, for each $j$ we have $\bigA{0}{j}{0} = 0$
  and we can compute in time $O(1)$ the value $\bigA{i}{j}{t}$
  whenever at least one of the arguments is $0$. We also adopt the
  convention that $\bigA{i}{j}{t} = \infty$ whenever at least one of
  the arguments is negative.  Further, for every pair of integers $i,i'$ with
  $i \leq i'$ and every candidate $c_j$, we define $\pos(i,i',c_j)$ to be
  the worst position of candidate $c_j$ in the votes that belong to the set
  $\{v_t \mid \max\{1,i\} \leq t \leq \min\{n, i'\} \}$.
  
  It is easy to see that the following equality holds for all
  $i\in[n]$, $j\in[m]$, $t\in[k]$
  (we take the minimum over an empty set to be $\infty$):
  \begin{align*}
  \bigA{i}{j}{t} = \min_{j' \leq j}\bigg\{\min\Big\{ 
  &\max\big\{ \alpha(\pos(i-\left\lfloor\frac{n}{k}\right\rfloor+1, i, c_{j'})), 
  \bigA{i-\left\lfloor\frac{n}{k}\right\rfloor}{j'-1}{t-1}\big\},\\
  &\max\big\{ \alpha(\pos(i-\left\lceil\frac{n}{k}\right\rceil+1, i, c_{j'})),
  \bigA{i-\left\lceil\frac{n}{k}\right\rceil}{j'-1}{t-1}\big\} \Big\}
\bigg\}.
  \end{align*}

  Using this equality and standard dynamic programming, it is easy to
  compute $\bigA{n}{m}{k}$ and the $k$-Monroe-assignment function that
  achieves this dissatisfaction in polynomial time. To be more
  precise, computing $\bigA{n}{m}{k}$ requires time
  $O(nm^2k)$;   to achieve this running time, we first compute
  all the necessary values
  $\pos(i-\left\lfloor\frac{n}{k}\right\rfloor+1,i,c_j)$ and
  $\pos(i-\left\lceil\frac{n}{k}\right\rceil+1,i,c_j)$ in time
  $O(nm)$. If $k$ divides $n$, the running time can be improved 
    to $O(nm^2)$ because we can omit the parameter $t$ in
    $\bigA{i}{j}{t}$.  
  Optimality of the computed assignment follows by
  Lemma~\ref{lem:n-cons-contig}.
\end{proof}

In a way, Theorem~\ref{thm:monroe-n-cons-easy} is not new: It can be
shown that narcissistic elections are necessarily single-peaked (this
is implicit in the work of Barber\`a and
Moreno~\cite{bar-mor:j:top-monotonicity}), and for single-peaked
elections Betzler et al~\cite{bet-sli-uhl:j:mon-cc} provide a
polynomial-time algorithm for $\ell_\infty$-Monroe (Proposition~5
in~\cite{bet-sli-uhl:j:mon-cc}).  Thus, if we only care about
polynomial-time computability, Theorem~\ref{thm:monroe-n-cons-easy}
does not appear to be useful. However, there are two reasons to prefer
the algorithm described in Theorem~\ref{thm:monroe-n-cons-easy}.
First, our algorithm is considerably faster: the running time of
Betzler et al.'s algorithm is $O(n^3m^3k^3)$, while for our algorithm
it is $O(nm^2k)$.  Second, our algorithm produces an assignment that
has the contiguous blocks property.  In contrast, in
Section~\ref{sec:contig} we show that this is not necessarily the case
for the algorithm of Betzler et al.

\section{Contiguous Blocks Property: Counterexamples}\label{sec:contig}
We will now explore the limitations of the algorithmic approach
that is based on the contiguous blocks property.

First, we provide an example of a single-crossing
election %
such that all optimal assignments under Monroe's rule do not possess
the contiguous blocks property; in fact, the approximation ratio of
every algorithm that produces assignments with this property is
$\Omega(m)$. This result holds both for the utilitarian and for the
egalitarian version of Monroe's rule.

\begin{example}\label{ex:non-contig-sc}
{\em
  Let $C = \{c_1, c_2\} \cup A \cup B$ be a set of candidates, where $\|A\| = \|B\| = m$.
  Let $V = V_1 + V_2 + V_3 + V_4$ be a list of voters, where
  each $V_i$ contains $n$ voters. Voters in each $V_i$ have
  identical preference orders, defined as follows.
\begin{align*}
  V_1: c_1 \succ \ordset{B} \succ c_2 \succ \ordset{A}, \\
  V_2: c_1 \succ c_2 \succ \revset{B} \succ \ordset{A}, \\
  V_3: c_1 \succ c_2 \succ \ordset{A} \succ \revset{B}, \\
  V_4: c_1 \succ \revset{A} \succ c_2 \succ \revset{B}.
  \end{align*}
  We seek a $2$-member parliament, and use Borda misrepresentation function $\alphaborda$.

  If we partition voters into contiguous blocks, then the best we can do is to assign $c_1$
  to $V_1+V_2$ and $c_2$ to $V_3+V_4$ (or the other way round), achieving
  dissatisfaction of $n(m+2)$ and $m+1$ under $\alphaborda$-$\ell_1$-Monroe and $\alphaborda$-$\ell_\infty$-Monroe, 
  respectively.
  In contrast, if we assign $c_1$ to $V_1+V_4$ and $c_2$ to $V_2+V_3$, then
  the total dissatisfaction under $\alphaborda$-$\ell_1$-Monroe and $\alphaborda$-$\ell_\infty$-Monroe 
  is $2n$ and $1$, respectively.
}
\end{example}

Further, even for narcissistic single-crossing elections, if we
consider the utilitarian version of Monroe's rule, imposing the
contiguous blocks property may lead to suboptimal solutions.

\begin{example}\label{ex:non-contig-util} {\em We consider the set $C
    = \{a,b,c,d,e,f\}$ of candidates and the following twelve voters
    (the profile is clearly narcissistic and it is easy to check that
    it is single-crossing): 
    \newcommand{\sixvote}[6]{ #1 \succ #2 \succ #3 \succ #4 \succ #5 \succ #6}
    \newcommand{\lowest}{\mathrm{lowest}}
    \allowdisplaybreaks \begin{align*}
      v_1&: \sixvote abcdef, \\
      v_2&: \sixvote bacdef, \\
      v_3&: \sixvote bcadef, \\
      v_4&: \sixvote cbadef, \\
      v_5&: \sixvote cbadef,\\
      v_6&: \sixvote cbadef, \\
      v_7&: \sixvote cbdaef, \\ %
      v_8&: \sixvote cdbaef,\\
      v_9&: \sixvote defcba,\\
      v_{10}&: \sixvote efdcba,\\
      v_{11}&: \sixvote efdcba,\\
      v_{12}&: \sixvote fedcba. \end{align*} 
    We seek a $2$-member parliament, and our voting rule
    is $\alpha_\borda$-$\ell_1$-Monroe.  If we require the assignment
    function to have the contiguous blocks property, then the
    unique optimal solution assigns $b$ to each voter in $V_1 = (v_1,
    \ldots, v_6)$ and $d$ to each voter in $V_2 = (v_7, \ldots, v_{12})$.
    Under this assignment the total misrepresentation 
    of the voters in $V_1$ and $V_2$ is given by $4$ and $9$, respectively, 
    so the optimal total misrepresentation for assignments with
    the contiguous blocks property is $13$.

    On the other hand, without the contiguous blocks property, 
    we can assign candidate $e$ to voters $v_1, v_2, v_9, v_{10}, v_{11}, 
    v_{12}$, and candidate $c$ to voters $v_3, v_4, v_5, v_6, v_7, v_8$, for
    the total misrepresentation of $11$. Indeed, we will now argue 
    that this is the unique optimal solution.
 
    For each candidate $x$, set
    $
    \lowest(x)=\min\left\{\sum_{i\in I}\alpha_\borda(\pos_{v_i}(x))\mid I\subseteq [12], \|I\|=6\right\}.
    $
    Intuitively, $\lowest(x)$ corresponds to assigning $x$ in the best possible way.
    We have 
    \begin{align*}
      &\lowest(a) = 9,&  &\lowest(b) = 4,&  &\lowest(c) = 1 \\
      &\lowest(d) = 9,& &\lowest(e) = 10,& &\lowest(f) =
      14.  
    \end{align*} 
    Thus each assignment function that gives total
    misrepresentation of at most $11$ produces one of the following sets
    of winners: 
    $\{a, c\}$, $\{b,c\}$, $\{c,d\}$, or $\{c,e\}$. 
    
    Now, if the set of winners is $\{a, c\}$ or $\{b, c\}$, 
    the total misprepresentation is at least $12$, because
    for each voter in $(v_9, v_{10}, v_{11}, v_{12})$
    assigning a candidate from either of these sets
    results in a misrepresentation of at least $3$.
    
    Consider the candidates $c$ and $d$.
    If each voter were
    assigned to her more preferred candidate among these two, we would get a
    $2$-CC-assignment function with misrepresentation $11$. However,
    under this assignment $8$ voters are assigned to $c$, so the
    best Monroe assignment for these two candidates has
    misrepresentation at least $12$. Finally, it is clear that no
    assignment function that uses $c$ and $e$ can have a lower
    misrepresentation than $11$ and we have seen that such a
    $2$-Monroe assignment exists.

    To summarize, imposing the contiguous blocks property may lead to
    a suboptimal assignment whose set of winners is disjoint from the
    one for the optimal assignment, even in a narcissistic
    single-crossing election.  }
\end{example}

We now consider single-peaked preferences.  Note that the definition
of a single-peaked election does not impose any restrictions on the
voter ordering, and therefore it is not immediately clear what is the
correct way to extend the contiguous blocks property to this
setting. However, it seems natural to order the voters according to
their most preferred candidates (using the candidate order given by
the axis), breaking ties according to their second most preferred
candidate, etc.  That is, consider an election $E=(C, V)$ that is
single-peaked with respect to the order $c_1 \lhd\dots \lhd c_m$. A
voter with a preference order $c_{i_1}\succ\dots\succ c_{i_m}$ can be
identified with the string $i_1\dots i_m$.  We reorder the voters so
that if $i<j$ then the string associated with $v_i$ is
lexicographically smaller than or equal to the string associated with
$v_j$.

We then ask if for every single-peaked election there exists an
optimal assignment for Chamberlin--Courant's rule or Monroe's rule
that satisfies the contiguous blocks property with respect to this
ordering.  We will now show that the answer is ``no'', both for the
egalitarian and the utilitarian version of both rules. Indeed, just as
in Example~\ref{ex:non-contig-sc}, imposing the contiguous blocks
property has a cost of $\Omega(m)$.

\begin{example}\label{ex:non-contig-sp}
{\em
Let $C = \{x_1, \dots, x_m, y_1, \dots, y_m, a, b, c, d\}$.
Let $V=(v_1, v_2, v_3, v_4)$, where 
\begin{align*}
  v_1: a\succ x_1\succ\dots\succ x_m \succ b \succ c\succ d \succ y_1\succ\dots\succ y_m, \\
  v_2: b \succ c\succ d \succ y_1\succ\dots\succ y_m \succ a\succ x_1\succ\dots\succ x_m, \\ 
  v_3: c\succ b\succ a\succ x_1\succ\dots\succ x_m \succ d \succ y_1\succ\dots\succ y_m, \\
  v_4: d \succ y_1\succ\dots\succ y_m \succ c\succ b\succ a\succ x_1\succ\dots\succ x_m.
  \end{align*}
It is easy to see that $E=(C, V)$ is single-peaked with respect to the axis
\[
x_m\succ\dots\succ x_1 \succ a \succ b \succ c\succ d \succ y_1\succ\dots\succ y_m.
\]
In fact, it can be shown that all axes witnessing that $E$ is single-peaked 
have the property that $a$, $b$, $c$, and $d$ appear in the center of the axis, 
ordered as $a>b>c>d$ or $d>c>b>a$ 
(this is implied, e.g., by the analysis in~\cite{esc-lan-ozt:c:single-peaked-consistency}).
Thus the ordering of the voters induced by the axis is $(v_1, v_2, v_3, v_4)$.
We seek a $2$-member parliament, and use Borda misrepresentation function $\alphaborda$.

Suppose that we assign $a$ to $v_1$ and $v_3$ and $d$ to $v_2$ and $v_4$.
Then the total dissatisfaction under $\alphaborda$-$\ell_1$-Monroe and $\alphaborda$-$\ell_\infty$-Monroe
is $4$ and $2$, respectively. However, if we impose the contiguous blocks property, 
then the optimal assignment for Monroe's is to assign $b$ to $v_1$ and $v_2$,
and $c$ to $v_3$ and $v_4$; this results in total dissatisfaction of $2(m+1)$
and $m+1$ with respect to $\alphaborda$-$\ell_1$-Monroe and $\alphaborda$-$\ell_\infty$-Monroe,
respectively. For $\alphaborda$-$\ell_1$-CC, we can obtain a somewhat better
dissatisfaction than for $\alphaborda$-$\ell_1$-Monroe, by assigning $b$ to $v_1$, $v_2$,
and $v_3$, and assigning $d$ to $v_4$. However, even for Chamberlin--Courant's rule
under every assignment that has the contiguous blocks property, 
at least one voter will be assigned to a candidate
that she ranks in position $m+2$ or lower, so the total (egalitarian or utilitarian)
dissatisfaction will be at least $m+1$.
}
\end{example}

Example~\ref{ex:non-contig-sp} illustrates that our algorithms for
Chamberlin--Courant's rule (for single-crossing elections) and for the
egalitarian version of Monroe's rule (for single-crossing narcissistic
elections) are quite different from the algorithms for these problems
proposed by Betzler et al.~\cite{bet-sli-uhl:j:mon-cc} for
single-peaked elections: even though all of these algorithms rely on
dynamic programming, the dynamic program for single-peaked elections
is very different from the one for single-crossing elections.  As a
consequence, if we apply Betzler et al.'s algorithm to a
single-crossing narcissistic election, we may fail to get an
assignment that satisfies the contiguous blocks property.
Examples~\ref{ex:non-contig-sc} and~\ref{ex:non-contig-util} further
show that we cannot hope to extend
Theorem~\ref{thm:monroe-n-cons-easy} to all single-crossing elections,
or to utilitarian preferences.

\section{Conclusions}\label{sec:conclusions}
We have considered the complexity of winner determination under
Chamberlin--Courant's and Monroe's rules, for the case of
single-crossing profiles.  We have presented a polynomial-time
algorithm for Chamberlin--Courant's rule for single-crossing elections
(and for elections that are close to being single-crossing in the
sense of having bounded single-crossing width), and an $\np$-hardness
proof for Monroe's rule for the same setting.
Our results further strengthen the intuition that Monroe's rule is
algorithmically harder than Chamberlin--Courant's rule.  Similar
conclusions follow from the work of Betzler et
al.~\cite{bet-sli-uhl:j:mon-cc} and Skowron et
al.~\cite{sko-fal-sli:c:multiwinner}

Inspired by our negative result for Monroe's rule, we have sought
further natural restrictions on voters' preferences.  To this end, we
considered single-crossing narcissistic profiles and developed an
efficient algorithm for the egalitarian version of Monroe's rule under
this preference restriction. However, we showed that our approach does
not extend to general single-crossing elections, or to the utilitarian
version of Monroe's rule.

Perhaps the most obvious direction for future research that is
suggested by our work is understanding the computational complexity of
the utilitarian version of Monroe's rule for single-crossing
narcissistic elections and of egalitarian version of Monroe's rule for
single-crossing profiles. While we have shown that approaches based on
the contiguous blocks property are bound to fail, other approaches may
be more successful.  Going in another direction, perhaps it is
possible to obtain efficient algorithms for our restricted domains and
for dissatisfaction functions other than Borda.

\bibliographystyle{alpha}
\bibliography{grypiotr2006}

\end{document}